\documentclass[10pt]{IEEEtran}
\usepackage{amsmath, mathrsfs,amsfonts}
\usepackage{amssymb}
\usepackage{graphicx}
\usepackage[letterpaper]{geometry}
\def\bs{\boldsymbol}

\def \foral {\textrm{for all }}

\usepackage{hyperref}

\newtheorem{theorem}{Theorem}
\newtheorem{corollary}{Corollary}
\newtheorem{lemma}{Lemma}
\newtheorem{remark}{Remark}

\newtheorem{assumption}{Assumption}
\newtheorem{definition}{Definition}

               {\begin{list}{}{\leftmargin#1\rightmargin#2\topsep#3}\item{}}%
               {\end{list}}

\title{Collision Helps\\ Algebraic Collision Recovery for Wireless Erasure Networks}

\author{Ali ParandehGheibi, Jay Kumar Sundararajan, Muriel M\'edard \\  Department of Electrical Engineering and
Computer Science\\ Massachusetts Institute of Technology  \\ Email:\{parandeh, jaykumar,
medard\}@mit.edu}

\begin{document}
\maketitle

\begin{abstract}
Current medium access control mechanisms are based on collision avoidance and collided packets are
discarded. The recent work on ZigZag decoding departs from this approach by recovering the original
packets from multiple collisions. In this paper, we present an algebraic representation of
collisions which allows us to view each collision as a linear combination of the original packets.
The transmitted, colliding packets may themselves be a coded version of the original packets.

We propose a new acknowledgment (ACK) mechanism for collisions based on the idea that if a set of
packets collide, the receiver can afford to ACK exactly one of them and still decode all the
packets eventually. We analytically compare delay and throughput performance of such collision
recovery schemes with other collision avoidance approaches in the context of a single hop wireless
erasure network. In the multiple receiver case, the broadcast constraint calls for combining
collision recovery methods with network coding across packets at the sender. From the delay
perspective, our scheme, without any coordination, outperforms not only a ALOHA-type random access
mechanisms, but also centralized scheduling. For the case of streaming arrivals, we propose a
priority-based ACK mechanism and show that its stability region coincides with the cut-set bound of the packet erasure network.
\end{abstract}

\section{Introduction}
The nature of the wireless network is intrinsically different from the wired network because of the
sharing of the medium among several transmitters. Such a restriction generally has been managed
through forms of scheduling algorithms to coordinate access to the medium, usually in a distributed
manner. The conventional approach to the Medium Access Control (MAC) problem is contention-based
protocols in which multiple transmitters simultaneously attempt to access the wireless medium and
operate under some rules that provide enough opportunities for the others to transmit. Examples of
such protocols in packet radio networks include ALOHA, MACAW, CSMA/CA, etc\cite{book:datanetworks}.

However, in many contention-based protocols, it is possible that two or more transmitters transmit
their packet simultaneously, resulting in a \emph{collision}. The collided packets are considered
useless in the conventional approaches. There is a considerable literature on extracting partial
information from such collisions. Gollakota and Katabi \cite{ZigZag} showed how to recover multiple
collided packets in a 802.11 system using \emph{ZigZag decoding} when there are enough
transmissions involving those packets. In fact, they suggest that each collision can be treated as
a linearly independent equation of the packets involved. ZigZag decoding is based on interference
cancelation, and hence, requires a precise estimation of channel attenuation and phase shift for
each packet involved in a collision. ZigZag decoding provides a fundamentally new approach to
manage interference in a wireless setting that is essentially decentralized, and can recover losses
due to collisions. In this work, we wish to understand the effects of this new approach to
interference management in the high SNR regime, where interference, rather than noise, is the main
limit factor for system throughput.

%We study the performance of collision recovery methods in terms of the achievable throughput and
%delay for a single-hop wireless network.

We provide an abstraction of a single-hop wireless network with erasures when a generalized form of
ZigZag decoding is used at the receiver, and network coding is employed at the transmitters. We
introduce an algebraic representation of the collisions at the receivers, and study conditions
under which a collision can be treated as a linearly independent equation (degree of freedom) of
the original packets at the senders.  We use this abstract model to analyze the delay and
throughput performance of the system in various scenarios.

First, we analyze a single-hop wireless erasure network, when each sender has one packet to deliver
to all of its neighbors. We characterize the expected time to deliver all of the packets to each
receiver when collisions of arbitrary number of packets are recoverable.  We observe that with
collision recovery we can deliver $n$ packets to a receiver in $n+O(1)$ time slots, where $n$ is
the degree of that particular receiver. This is significantly smaller than the delivery time of
centralized scheduling and contention-based mechanisms such as slotted ALOHA. In the case that
collisions of only a limited number of packets can be recovered, we propose a random access
mechanism in conjunction with ZigZag decoding to limit the level of contention at the receiver. Our
numerical results show that such a scheme provides a significant improvement upon contention-based
mechanisms even if each recoverable collision is limited to only two packets.

Second, we analyze the throughput of this system in a scenario where packets arrive at each sender
according to some arrival process. In this scenario, each sender \emph{broadcasts} a random linear
combination of the packets in its queue, and the receivers perform generalized form of ZigZag
decoding for interference cancellation. We characterize the stability region of the system, and
propose a decentralized acknowledgement mechanisms to stabilize the queues at the senders. The
stability region of the system with collision recovery achieves the cut-set outer bound of the
erasure network, that is \emph{strictly larger} than that of the system with centralized
scheduling.

The information theoretic capacity of wireless erasure network has been studied in the related
literature. The works by Dana \emph{et al.} \cite{hassibi}, Lun \emph{et al.}
\cite{desmond_reliable}, and Smith and Hassibi \cite{smith} focus on a wireless erasure network
with only broadcast constraints, while Smith and Vishwanath \cite{vishwanath} study the capacity of
an erasure network by considering only interference constraints. These works show how to achieve
the cut-set bound of the multi-hop erasure network under specific constraints for a single unicast
or multicast session. In contrast, our work takes into account both broadcast and interference
constraints, and studies the stability \emph{region}  for multiple sessions over a single-hop
wireless network. Another related literature investigates collision recovery methods such as the
works by Tsatsanis \emph{et al.} \cite{NDMA}, and Paek and Neely \cite{neely_zz}. In this
literature, once a collision of $k$ packets occurs, all senders remain silent until those involved
in the collision retransmit another $k-1$ times. Our proposed scheme, however, does not require
such coordination among the senders.

 The rest of this paper is organized as follows. In Section \ref{model_sec}, we present an
abstract model of a single-hop wireless network with erasures. Section \ref{abstraction_sec}
discusses an algebraic representation of the collisions at the receivers. Section \ref{delay_sec}
is dedicated to mean delivery time characterization of a single-receiver system for various
interference management schemes. In Section \ref{throughput_sec}, we characterize the stability
region of the single-receiver system with collision recovery. In Section \ref{multiple_sec}, we
generalize the results of preceding sections to the case of a single-hop wireless network with
multiple receivers. Finally, concluding remarks and extensions are discussed in Section
\ref{conclusion_sec}.

\section{System model}\label{model_sec}
The system consists of a single-hop wireless network with $n$ senders and $r$ receivers. We assume
that a node cannot be both a sender and a receiver. The connectivity is thus specified by a
bipartite graph. Fig. \ref{fig:bipartite} shows an example of such a network.

\begin{figure}
\centering
  \includegraphics[width=.4\textwidth]{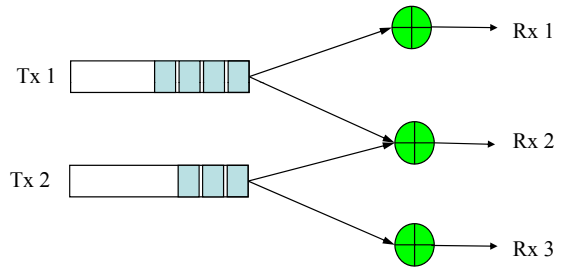}\\
  \caption{Single-hop wireless network model with $n$ senders and $r$ receivers}\label{fig:bipartite}
\end{figure}

We assume that time is slotted. Every sender is equipped with an infinite sized buffer. The goal of
a sender is to deliver all of its packets to each of its neighbors, i.e., the set of receivers to
which it is connected.

In every slot, a sender can \emph{broadcast} a packet to its neighbors. Owing to the fading nature
of the wireless channel, not all packet transmissions result in a successful reception at every
neighbor. Each link between any sender $i$ and any receiver $j$ may experience packet erasures.
These erasures occur with probability $p$, and are assumed to be independent across links and over
time.  This type of erasure is to model the effect of obstacles between the senders and the
receivers. The channel state between $i$ and $j$ is denoted by $c_{ij}(t)$.

At the end of every slot, each receiver is allowed to send an acknowledgment (ACK) to any one of
the senders to which it is connected. A packet is retained in the sender's queue until it has been
acknowledged by all the receivers. We ignore the overhead caused by the ACKs, and assume that the
ACKs are delivered reliably without any delay.

Note that a collision of packets at a receiver does not immediately imply an erasure. With ZigZag
decoding, it may be possible to extract useful information from collisions. In the following, we discuss how a collision could be thought of as a linear combination of the original
packets at the sender.

\subsection{An algebraic representation of collisions}\label{abstraction_sec}
In this section, we introduce an algebraic representation of collisions. The collision of two
packets is essentially the superposition of the physical signal corresponding to the packets. A packet is essentially a vector of bits that can be grouped into symbols over a finite field $\mathbb{F}_q$. For the rest of this section, we represent a packet as a polynomial over the delay variable $D$, with coefficients being the symbols of $\mathbb{F}_q$ that form the packet. The mapping from the packet to the corresponding physical signal is a result of two operations -- channel coding and modulation. We abstract these two operations in the form of a map $M$ from symbols over $\mathbb{F}_q$ to the complex number field:
\[M: \mathbb{F}_q \rightarrow \mathbb{C}\]
We assume that the map $M$ is such that given a complex number, there is a well-defined demodulation and channel decoding method that outputs the symbol from $\mathbb{F}_q$ that is most likely to have been transmitted.

\begin{remark}
The above assumption essentially says that the channel coding occurs over blocks of $\log_2 q$ bits
(corresponding to a single symbol of $\mathbb{F}_q$). Depending on $q$, this could mean a short
code length, which would be effective only with a high SNR.
\end{remark}

Let $X(D)$ and $Y(D)$ be two packets at two different senders, represented as polynomials over $\mathbb{F}_q$. The coding and modulation results in a signal polynomial over the complex field: $S_X(D)$ and $S_Y(D)$. Now imagine that these two packets collide with each other at a receiver twice, in two different time slots. We denote $h^{(t)}_j$ to be the channel coefficient in slot $t$ from sender $j$.

When packets collide, they may not be perfectly aligned. Let $u^{(t)}_{j}$ denote the offset (in symbols) of the packet from sender $j$ within slot $t$ measured from the beginning of the slot. We assume that a packet is significantly longer than the offsets so that the loss of throughput because of these offsets is negligible.

The channel gains, offsets and the identity of the packets that are involved in the collision are assumed to be known at the receiver. Then, the two collisions can be represented in the following way:
\[\left(\begin{array}{c}C_1(D)\\C_2(D)\end{array}\right) = \left(
\begin{array}{cc}
h^{(1)}_{1}D^{u^{(1)}_1} & h^{(1)}_{2}D^{u^{(1)}_2}\\
h^{(2)}_{1}D^{u^{(2)}_1} & h^{(2)}_{2}D^{u^{(2)}_2}\\
\end{array}
\right)\left(\begin{array}{c}S_1(D)\\S_2(D)\end{array}\right),\]
\noindent or alternately, $C=HS$.

Therefore, with $n$ collisions of the same $n$ packets, it is possible to decode them all as long as the $n\times n$ transfer matrix $H$ is invertible over the field of rational functions of $D$. The process of decoding by inverting this matrix is more general than the ZigZag procedure of \cite{ZigZag}. The decoding process will result in the signals corresponding to the original packets. The signals will then have to be demodulated and decoded (channel coding) to obtain the original data. This algebraic representation formalizes the intuition introduced in \cite{ZigZag} that every collision is like a linear equation in the original packets.

%In this manner, the signal polynomial of a collision of packets can be represented as a linear combination of the signal polynomials of the original packets, with coefficients chosen from the field of rational functions of $D$, defined over the complex field.

\subsection{Combining packet coding with collision recovery}
Due to the broadcast constraint of the wireless medium, a sender that wants to broadcast data to several receivers will have to code across packets over a finite field in order to achieve the maximum possible throughput. Random linear coding is known to achieve the multicast capacity over wireless erasure networks \cite{desmond_reliable}. Let us suppose that the sender codes across packets over the field $\mathbb{F}_q$ and that the coding coefficients are known at the receiver.

This can also be incorporated into the above formulation in the following sense. Suppose a receiver receives $n$ collisions, where the colliding packets in each collision are themselves finite-field linear combinations of a collection of $n$ original packets, then it is possible to decode all $n$ packets from the collisions.

\begin{figure}
\centering
  \includegraphics[width=.8\columnwidth]{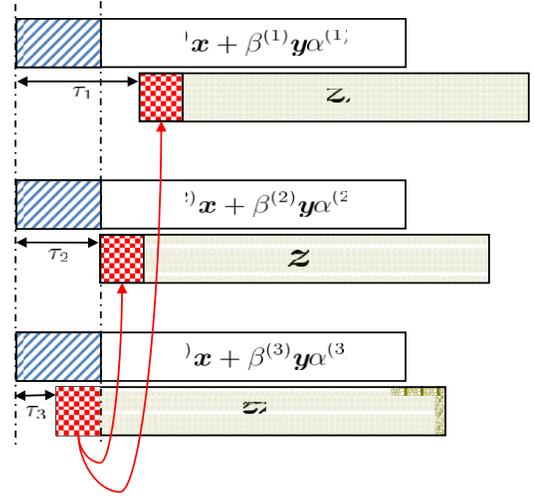}\\
  \caption{Combining packet coding with collision recovery}\label{ZigZag}
\end{figure}

This is immediately seen if we assume that that the coding and modulation are linear operations, i.e., that $M$ is a linear function with respect to the symbols of the original packets. In this case, the above matrix representation will still hold, and the invertibility condition for decoding will also be true. However, in general, the modulation operation may not be linear with respect to the original packets' symbols. Even in this case, we can still decode the $n$ packets from $n$ collisions.

We explain this using a simple example with two senders and one receiver. Suppose the first sender has two packets $\mathbf{x}$ and $\mathbf{y}$ and the second sender has a single packet $\mathbf{z}$. The first sender transmits a random linear combination of its two packets in every slot, while the second sender repeat packet $\mathbf{z}$ in every slot. Figure \ref{ZigZag} shows the collisions in three different time slots. Using the three collisions, the receiver can decode all three packets as follows. The offsets between the first and second senders' packets in the three collisions are $\tau_1$, $\tau_2$ and $\tau_3$. From the figure, since the first $\tau_2$ symbols of the first two collisions are interference-free, we can decode the first $\tau_2$ symbols of $\mathbf{x}$ and $\mathbf{y}$. Using this, we can compute the first $\tau_2$ symbols of $\alpha^{(3)} \mathbf{x}+\beta^{(3)} \mathbf{y}$, and thereby obtain the first $(\tau_2-\tau_3)$ symbols of $\mathbf{z}$. This !
 process can be continued after subtracting these symbols from the other collisions.

We assume throughout this paper that the field size $q$ is large enough that every collision counts as a new degree-of-freedom (also called \emph{innovative}) if and only if it involves at least one packet that has not yet been decoded. Every such collision counts as one step towards decoding the packets.

%In summary, we can view the resulting collided signal as though it were a linear equation involving the original packets' signals.
%For the remaining part of the paper, we shall assume that a collision of linear combinations of the original packets is itself a linear equation involving the packets.

\section{Delivery time characterization for the single receiver case}\label{delay_sec}
In this section, we study a special case where there is only one receiver in the network. We shall
show later in this paper that the results derived in this section generalize to the multiple
receiver case. We study a scenario where every sender has a single packet that needs to be
delivered to the receiver.

\begin{definition}
Consider a single-hop network with a single receiving node and $n$ senders, each having one packet
to transmit. Define the \emph{delivery time}, $T_D(n)$, as the time to transmit all packets
successfully to the receiver.
\end{definition}

We can divide the delivery time into $n$ portions, where the $k^{th}$ portion corresponds to the
additional time required to for the receiver to send the $k^{th}$ ACK, starting from the time when
the previous (i.e. $(k-1)^{st}$) ACK was sent. We define the following notation, for $k=1,2,\ldots
n$:

\

\begin{tabular}{lcp{2.5in}}
\ \ $T_k$ &=& Time when the receiver sends the $k^{th}$ acknowledgment\\
\vspace{.1in}
\ \ $X_k$ &=& $T_k-T_{k-1}$\ \ \ \ ($T_0$ is assumed to be 0).\\
\end{tabular}

Note that $T_D(n)$ is then given by:
\begin{equation}\label{eq:tdx}
T_D(n)=T_n=\sum_{k=1}^n X_k
\end{equation}

The goal of this section is to characterize the expectation of the delivery time for ZigZag
decoding, and to compare it with contention-based protocols and a central scheduling mechanism.

First we study schemes that treat any collision as a loss. In this case, collisions have to be
avoided either by centralized coordination among the senders, or in a distributed way by having
senders access the channel in a probabilistic manner, as studied in the literature (Please see
Chapter 4 of \cite{book:datanetworks} for a summary).

\subsection{Centralized scheduling}
We assume that the receiver, upon successfully receiving a packet, sends an acknowledgment to the
corresponding sender. With centralized scheduling, we assume the following policy. The channel is
initially reserved for sender 1, up to the point when its packet is acknowledged. At this point,
the channel is reserved for sender 2, and so on. In this setting, the calculation of the expected
delivery time is straightforward. For each sender, the delivery is complete in the first slot when
the channel from that sender to the receiver is not under erasure. The time $X_k$ between the
$(k-1)^{st}$ and the $k^{th}$ ACK, which is also the delivery time for the $k^{th}$ sender, is thus
a geometric random variable, with mean $\frac {1}{1-p}$. This implies that the total expected
delivery time under centralized scheduling policy is given by:
\[\mathbb{E}[T_D(n)]=\sum_{k=1}^n \mathbb{E}[X_k]= \frac{n}{1-p}.\]

Note that the delivery time for centralized scheduling is normally a lower bound for the delivery
time of other distributed probabilistic approaches because it ensures that there is no collision.
In distributed approaches, there is always some probability of a collision.

\subsection{Random access}
In this case, we assume that in every slot, each sender transmits its packet with probability $q$
until it is acknowledged. The choice of whether to transmit or not is made independently across
senders and across time. Note that, by controlling the access probability $q$, the senders can
control the level of contention and thereby prevent collisions.

\begin{theorem}\label{thm:randaccess}
  The expected delivery time for the random access scheme with an access probability $q$ is given by:
 \[\mathbb{E}[T_D(n)]=\sum_{k=1}^n\frac{1}{kq_e(1-q_e)^{k-1}}.\]
where $q_e = q(1-p)$ is the effective probability of access, after incorporating the erasures.
\end{theorem}
\begin{proof}
If a sender decides to transmit in a given slot, then it might still experience an erasure with
probability $p$. Hence, the effective access probability of a sender is given by $q_e = q(1-p)$.

Consider the interval corresponding to $X_{n-k+1}$. In this interval, there are $k$ unacknowledged
senders. Therefore, at each time slot, the number of senders that the receiver can hear from
follows a binomial distribution with parameters $(k, q_e)$. A successful reception occurs when
exactly one sender is connected, which happens with probability $kq_e(1-q_e)^{k-1}$. Thus,
$X_{n-k+1}$ is a geometric random variable with mean $(kq_e(1-q_e)^{k-1})^{-1}$. The result follows
from Eqn. (\ref{eq:tdx}).
\end{proof}

\begin{corollary}
By selecting the access probability $q=\frac1n$, we get \[\mathbb{E}[T_D(n)]=O(n \log n).\]
\end{corollary}

\subsection{ZigZag decoding}
Next, we consider the scenario where the receiver has ZigZag decoding capability. In this scenario,
every sender transmits its packet in every slot until acknowledged by the receiver.

With ZigZag decoding, there are multiple ways to acknowledge a packet. The conventional method is
to ACK a packet when it is decoded. However, we propose a new ACK mechanism that is not based on
decoding. The key observation is that upon receiving an equation (collision), the receiver can
afford to ACK \emph{any one} of the senders involved in that collision.

In the following theorem, we show that this form of acknowledgments will still ensure that every
packet is correctly decoded by the receiver eventually.

\begin{theorem}\label{thm:correctness}
  Consider a single-hop network with $n$ senders and one receiver capable of performing ZigZag
decoding. Suppose the receiver, upon a reception, acknowledges an arbitrary sender among those
involved in the collision. At the point when the receiver sends the $n^{th}$ ACK, it can
successfully decode all $n$ packets.
\end{theorem}
\begin{proof}
  Let $D_k$ be the set of packets that have been decoded at time $T_k$, i.e., immediately after
sending the $k^{th}$ ACK. Also, let $A_k$ be the set of packets that have been ACKed at time $T_k$
including the $k^{th}$ ACK. We shall show that $D_k \subseteq A_k$ for all $k=1, 2, \ldots n$.

For any $k=1, 2, \ldots n$, let $|D_{k}|=m$. This means, among the first $k$ receptions, there are
at least $m$ linearly independent equations involving only these $m$ packets (from Section
\ref{abstraction_sec}). For every reception, the receiver always ACKs exactly one of the senders
involved in the collision. This means, corresponding to these $m$ equations, $m$ ACKs were sent by
the receiver to a set of senders within $D_{k}$.

 An ACKed sender never transmits again. Since the receiver always ACKs one of the senders involved
in a collision, no sender will be ACKed more than once. Hence, these $m$ ACKs are sent to $m$
distinct senders in $D_{k}$. This means all senders in $D_k$ have been ACKed.

  We have shown that $D_k \subseteq A_k$ for all $k=1, 2, \ldots n$. A sender that has been ACKed
will not transmit again. Hence, every reception will only involve senders whose packet has not been
decoded. This implies that every reception is innovative, since a reception is innovative if and
only if it involves at least one sender whose packet has not yet been decoded (see Section
\ref{abstraction_sec}).

  Therefore, at the point of sending the $n^{th}$ ACK, the receiver has $n$ linearly independent
equations in $n$ unknowns, and hence can decode all the packets.
\end{proof}

We shall now derive the expected delivery time for ZigZag decoding.

\begin{theorem}\label{thm:delay_zz}
For ZigZag decoding, the expected delivery time is given by:
  \[\mathbb{E}[T_D(n)] = \sum_{k=1}^n \frac 1{1-p^k} = n+O(1).\]
\end{theorem}
\begin{proof}
  At time $T_k$, $k$ distinct senders have been ACKed, and only $(n-k)$ senders will attempt
transmission. From the proof of Theorem \ref{thm:correctness}, every collision at the receiver will
result in an innovative linear combination. Hence, an innovative reception occurs if and only if
not all of the $(n-k)$ senders experience an erasure. The time to receive the next innovative
packet, $X_{k+1}$, is thus a geometric random variable with mean $1/(1-p^{n-k})$. Now, by Eqn.
\ref{eq:tdx}, we obtain the following:

  \begin{eqnarray*}
    \mathbb{E}[T_D(n)]&=&\sum_{k=1}^n \frac 1{1-p^k} = n+\sum_{k=1}^n\frac{p^k}{1-p^k}\\
    &\le &n+\frac1{1-p}\sum_{k=1}^n p^k \le n+\frac{p}{(1-p)^2} =n+O(1).
  \end{eqnarray*}
\end{proof}

% Compare with centralized scheduling
Let us now compare this scheme with a centralized scheduling mechanism. Centralized scheduling
requires a central controller that assigns every time slot to a single sender, and achieves a
delivery time of $n/(1-p)$. In contrast, in the ZigZag-based approach, no coordination is necessary
among the senders, and yet, the delivery time is $n+O(1)$, that is close to the lowest possible
time of $n$ slots, required to deliver $n$ packets.

Such an improvement in performance can be explained as follows. For centralized scheduling, since
only one user is scheduled to transmit in a time-slot, the time-slot will be wasted from the
receiver's point of view, with probability $p$. In contrast, with ZigZag, since all the
unacknowledged senders attempt to access the channel in a given slot, we obtain a \emph{diversity}
benefit -- if even one of the attempting senders does not experience an erasure, the slot is useful
to the receiver.
\subsection{ZigZag decoding with random access}
The earlier subsection assumed that a collision of any number of packets can be treated as a linear
equation involving those packets. The largest number of packets that can be allowed to collide for
ZigZag decoding to still work depends on the range of the received Signal-to-Noise Ratio (SNR). In
practice, if a collision involves more than 3 or 4 packets, then the ZigZag decoding process is
likely to fail, owing to error propagation.

Hence, in a more realistic setup, we need to limit the level of contention in order to ensure that
more collisions at the receiver are useful. In this part of the paper, we explore the possibility
of combining ZigZag decoding with random access. Instead of allowing every unacknowledged sender to
transmit, each sender opportunistically transmits its packet with some probability $q$. Thus, the
expected number of transmitting senders is reduced, which in turns limits the expected number of
colliding packets in one time slot. We assume that any collision involving more than $C$ packets is
not useful. This scheme is expected to perform better than conventional random access with no
ZigZag decoding, since a collision of $C$ or fewer packets is not useless, but is treated as one
received linear equation. Under this assumption, we can derive the expected delivery time in a
manner similar to the analysis of simple random access.

\begin{theorem}\label{thm:randaccess2}
The expected delivery time for the random access scheme with an access probability $q$ is given by:
$$\mathbb{E}[T_D(n)]=\sum_{k=1}^n\frac{1}{\sum_{m=1}^{\min(C,k)} {k\choose m}   q_e^m(1-q_e)^{k-m}},$$
where $q_e = q(1-p)$ is the effective probability of access, after incorporating the erasures.
\end{theorem}
\begin{proof}
  Consider the interval corresponding to $X_{n-k+1}$. In this interval, there are $k$ unacknowledged senders. Therefore, as in Theorem \ref{thm:randaccess}, at each time slot, the number of senders that the receiver can hear from follows a binomial distribution with parameters $(k, q_e)$, where $q_e$ is the effective access probability of a sender, given by $q_e = q(1-p)$.

  A successful reception occurs when $C$ or fewer senders is connected, which happens with probability \[p_k=\sum_{m=1}^{\min(C,k)} {k\choose m}  q_e^m(1-q_e)^{k-m}.\]
  Thus, $X_{n-k+1}$ is a geometric random variable with mean $1/p_k$. Using Eqn. \ref{eq:tdx}, we obtain the desired result.
\end{proof}

The design parameter $q$ should be chosen so as to minimize the delivery time. Unfortunately, the
exact characterization of the optimal $q$ in closed form seems difficult to obtain. In the
following section, we compare the expected delivery time for the above schemes, with the optimal
values of $q$ computed numerically.

\subsection{Numerical results}
\begin{figure}
\centering
  \includegraphics[width=.4\textwidth]{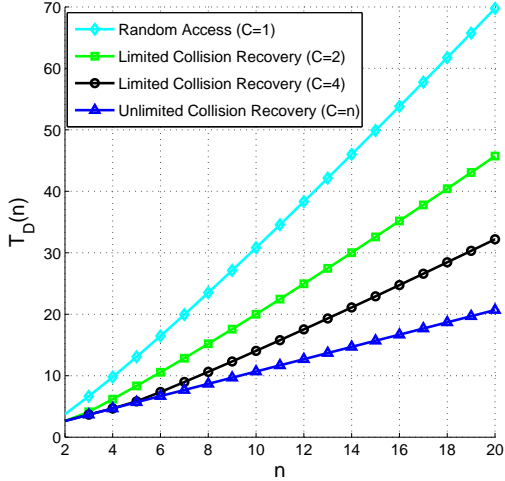}\\
  \caption{The delivery time for $p=1/3$ for different schemes}\label{delivery_zz}
\end{figure}

Fig. \ref{delivery_zz} shows the expected delivery time for the different schemes discussed above,
as a function of the number of senders $n$. The plot compares random access with ZigZag decoding
for different values of the contention limit $C$, which is the maximum number of packets that can
be allowed to collide for the collision to be considered useful.

The contention level is controlled by adjusting the access probability $q$. In the unlimited ZigZag
case, i.e., when we have no contention limit,  there is no need to reduce contention through random
access, and hence $q$ is set to 1. For the other cases, for each $n$, the value of $q$ is chosen so
as to minimize the delivery time.

The main observation is that by allowing ZigZag decoding, the expected delivery time is
significantly reduced, as compared to conventional random access where any collision is treated as
being useless.

We also  observe that the delivery time drops with an increase in the contention limit $C$. In the
unlimited ZigZag case, we can see that the delivery time is very close to the best possible time of
$n$ slots.

The value of the erasure probability $p$, is fixed at 1/3. However, we found that varying the value
of $p$ does not significantly affect the delivery time for the other schemes. In contrast, the plot
for the centralized scheduling case (not shown in the figure), would be a straight line with slope
$1/(1-p)$. In other words, the delivery time for centralized scheduling is sensitive to $p$.

Intuitively, the reason is, the random access approaches are allowed to change the access probability to reach a certain level of contention at the receiver. as the erasure probability $p$ increases, the senders can compensate by increasing their access probability $q$ to achieve the same contention level.

\section{Stability region for the single receiver case}\label{throughput_sec}
In this section, we consider a scenario when packets arrive at sender $i$ according to an arrival
process $A_i(t)$, where $A_i(t)$ represents the number of packets entering the $i^{th}$ sender's
queue at slot $t$ (cf. Fig. \ref{stream_fig}). We assume the arrival processes are
\emph{admissible} as defined in \cite{neely}.

\begin{assumption}\label{assump:arrival}
 The arrival processes satisfy the following conditions.
following conditions.
\begin{enumerate}
%  \item The arrival processes at different senders are independent.
  \item  $\lim_{t \rightarrow \infty} \frac1t \sum_{\tau = 0}^{t-1}\sum\mathbb E[A_i(t)] =
  \lambda_i$.
  \item There exists a finite value $A_{max}$ such that $\mathbb E[A_i^2(t)| \mathcal H(t)] \leq
  A_{max}^2$ for all $i$ and  $t$, where $\mathcal H(t)$ denotes the history up to time $t$.
  \item For any $\delta >0$, there exists an interval of size $T$ such that for any initial slot $t_0$
  $$\mathbb E\bigg[ \frac1T \sum_{\tau = 0}^{T-1} A_i(t_0 + \tau) | \mathcal H(t_0)\bigg] \leq \lambda_i +\delta, \quad \foral i.   $$
\end{enumerate}
\end{assumption}

The above conditions are easily satisfied if the arrival processes are Bernoulli processes with
mean $\lambda_i$. Let $\mu_i(t)$ be the number of packets dropped from the queue of the $i^{th}$
sender during time slot $t$. According to the communication protocol described in Section
\ref{model_sec}, a packet is dropped from a sender's queue if and only if it is acknowledged by all
the receivers connected to that sender. We also assume that the $A_i(t)$ arrivals occur at the end
of slot $t$. Thus, the evolution of $Q_i(t)$, the queue-length at sender $i$ at time $t$, is given
by
\begin{equation}\label{q_dynamic}
    Q_i(t+1) = \max \{Q_i(t) - \mu_i(t), 0\} + A_i(t).
\end{equation}

The goal is to characterize the stability region, which is defined as the closure of the set of
arrival rates for which there exist a service policy such that the each queue has a bounded time
average, i.e.,
$$\limsup_{t \rightarrow \infty} \frac1t \sum_{\tau = 0}^{t-1} \mathbb E[Q_i(\tau)] < \infty, \quad \foral i.$$

\begin{figure}
\centering
  \includegraphics[width=.45\textwidth]{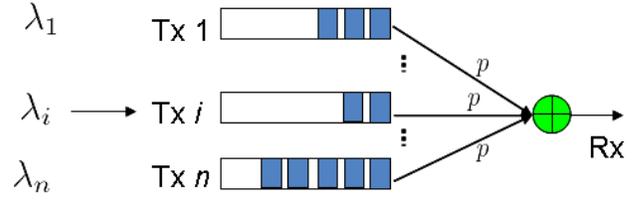}\\
  \caption{Single-hop network with $n$ senders and one receiver -- streaming case}\label{stream_fig}
\end{figure}

A  \emph{centralized scheduling policy} involves choosing at most one of the senders for transmission (service) so that any collision is avoided. If the packet is delivered successfully at the receiver, an acknowledgment is fed back to the sender and that packet is dropped from the sender's queue. The centralized scheduler requires coordination among the senders as well as information about the queue-length or the arrival rates. However, it does not have access to
channel state before it is realized. Therefore, probability of packet loss is independently at least $p$ at every time slot, and it is also independent of the implemented centralized scheduling policy. Thus, we have the following \emph{necessary} conditions for the stability region:
\begin{eqnarray}
   \sum_{i=1}^n \lambda_i &<& 1-p, \nonumber \\
    \lambda_i &\geq& 0, \quad i = 1, \ldots, n. \label{centralized_region}
\end{eqnarray}

In fact, it can be shown that the above conditions are also sufficient. The queues can be
stabilized by a centralized scheduling policy that selects the sender with the longest queue for
transmission \cite{neely}. In summary the stability region for centralized scheduling policies is
an $n$-dimensional simplex given by (\ref{centralized_region}). An example of such region for a
two-sender system is illustrated in Fig. \ref{region_fig}(a).

\begin{figure}
\centering
  \includegraphics[width=.45\textwidth]{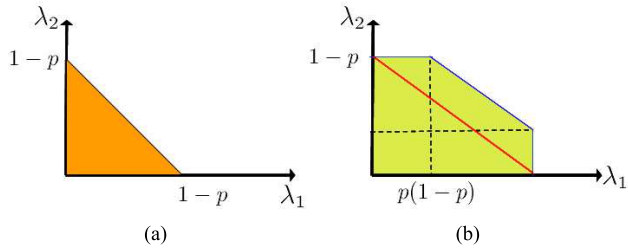}\\
  \caption{Stability region of a two-sender case with (a) centralized scheduling (b) ZigZag decoding.}\label{region_fig}
\end{figure}

Note that under the centralized scheduling policy the assigned sender may experience an erasure, and hence, wastes time slots even if there are other senders that would not have suffered an erasure. However, if the realization of the channel state in the next time slot is known, such wastes can be avoided by choosing the transmitter from those that are connected to the receiver. Tassiulas and Ephremides \cite{LCQ} show that if information about channel state realization is available a priori, the following set of arrival rates can be stabilized:
\begin{eqnarray}
   \sum_{i \in S} \lambda_i &<& 1-p^{|S|}, \quad \textrm{for all } S \subseteq \{1,\ldots, n\}, \nonumber \\
    \lambda_i &\geq& 0, \quad i = 1, \ldots, n, \label{MAC_region}
\end{eqnarray}
where $|S|$ denotes cardinality of set $S$. The region described in (\ref{MAC_region}) can be achieved by serving the sender with longest queue-length among those that are connected to the receiver. Moreover, Tassiulas and Ephremides \cite{LCQ} show that it is not possible to stabilize the queues for any point outside the region described in (\ref{MAC_region}). This can be seen as a consequence of \emph{Cut-Set bound} (cf. \cite{cover_book}) applied to this setup. The stability region for a two-sender system is illustrated in Fig. \ref{region_fig}(b).

In the following, we first show how to use ZigZag decoding scheme to achieve the dominant face of the stability region given in (\ref{MAC_region}) without prior knowledge about channel state realizations. We then show that as long as the sender side queues are stable, the receiver will eventually decode every packet that arrives at any sender.

\begin{definition}\label{priority_policy}
The \emph{priority-based} policy for a single-hop network with a single receiver is as follows. Fix a priority order of the senders with 1 being the highest priority.
\begin{itemize}
  \item \textbf{Transmission mechanism}: Each sender transmits the head-of-line packet of its queue at every time slot.
  \item \textbf{Acknowledgement mechanism}: Upon every reception, the receiver acknowledges the packet from
  the sender with highest priority among those packets that are involved in the collision.
  Consequently, each acknowledged packet is dropped from the corresponding sender's queue, i.e., $\mu_i(t) =
  1$ for sender $i$ if and only if $i$ is the highest priority sender with the following two properties: $Q_i(t) \neq 0$, and the link from sender $i$ to the receiver is
  not experiencing an erasure.
\end{itemize}

\end{definition}

In the following, we show the priority-based policy can achieve vertices of the stability region given by (\ref{MAC_region}). First, let us provide a simple characterization of the vertices of the dominant face of the region.

\begin{lemma}\label{vertex_lemma}
There exists a one-to-one correspondence between permutations of $\{1,\ldots,n\}$ and vertices of the dominant face of the region described in (\ref{MAC_region}). In particular, for any permutation $\pi$, the corresponding vertex is given by
$$\lambda_{\pi_i} = (1-p)p^{i-1}, \quad i = 1, \ldots, n.$$
\end{lemma}
\begin{proof}
See \cite{polymatroid_vertex}.
\end{proof}

\begin{theorem}\label{vertex_thm}
Consider a single-hop wireless erasure network with one receiver and $n$ senders, where the arrival
process $A_i(t)$ satisfies Assumption \ref{assump:arrival}. Any vertex on the dominant face of the
region given by (\ref{MAC_region}) can be achieved without prior knowledge about channel state
realization by employing the priority-based policy.
\end{theorem}
\begin{proof}
Fix a vertex, $V$, on the dominant face of the stability region. By Lemma \ref{vertex_lemma}, it
corresponds to a permutation $\pi$ of the senders. Without loss of generality, assume $\pi = (1, 2,
\ldots, n)$. By Lemma \ref{vertex_lemma}, the rate-tuple corresponding to $V$ is given by

\begin{equation}\label{vetex_rates}
\bar \lambda_{i} = (1-p)p^{i-1}, \quad i = 1, \ldots, n.
\end{equation}

Next, we show the priority-based policy defined in Definition \ref{priority_policy} can achieve the
vertex $V$, i.e., for any $\epsilon > 0$, the priority-based policy stabilizes the queues with
arrival rates
$$\lambda_{i} = (1-p)p^{i-1} - \epsilon, \quad i = 1, \ldots, n.$$

As we discussed in the acknowledgement mechanism of the priority-based policy, a sufficient
condition for acknowledging sender $i$ is to have the link of sender $i$ not erased and the links
of all other senders with higher priorities erased. Note that an acknowledgment to sender $i$ is
equivalent to serving the queue at sender $i$ by one packet. By independence of the erasures across
links we obtain the following expected service rate for each sender $i$
$$\mathbb E[\mu_i(t)] \geq p^{i-1} (1-p), \quad i = 1, \ldots, n.$$

Hence, by Definition 3.5 of \cite{neely}, the server process $\mu_i(t)$ is admissible with rate
$\bar \mu_i = p^{i-1} (1-p)$. Moreover, the arrival process $A_i(t)$ is also admissible with rate
$\lambda_i$ by Assumption \ref{assump:arrival}. Since $\bar \mu_i > \lambda_i$ for any $\epsilon
>0$, by Lemma 3.6 of \cite{neely} the sender side queues are stable. In other words, arrival rates
arbitrarily close to that of vertex $V$ can be achieved.
\end{proof}

\begin{corollary}
The dominant face of the stability region described in (\ref{MAC_region}) is achievable without
prior knowledge about channel state realization by employing the priority-based policy.
\end{corollary}
\begin{proof}
Every point on the dominant face of the stability region can be written as a convex combination of
the vertices of the dominant face. Moreover, each vertex can be achieved by a priority-based policy
given in Definition \ref{priority_policy}, corresponding to that vertex. Therefore, every point on
the dominant face can be achieved by time sharing between such policies. Note that the difference
between he policies achieving different vertices is in the acknowledgement mechanism which takes
place at the receiver, and no coordination among the transmitters is necessary.
\end{proof}

\begin{theorem}\label{decodability_single}
 For the priority-based policy, every packet that arrives at any sender will eventually get decoded
 by the receiver if it employs ZigZag decoding.
\end{theorem}
\begin{proof}
  By the same arguments as in the proof of Theorem \ref{thm:correctness}, every packet that is
  decoded at the receiver must have been acknowledged. If a packet is acknowledged, it is never
  transmitted again. Therefore, every collision at the receiver only involves packets that have not
  yet been decoded. Thus, every successful reception at the receiver is innovative (see Section
  \ref{abstraction_sec}).

In other words, the receiver sends out an ACK when and only when it receives an innovative packet.
This means that the total number of packets that have been dropped from any sender's queue at a
given time is equal to the total number of degrees of freedom at the receiver.

By Theorem \ref{vertex_thm}, the queue at each sender is stable. Hence, all the queues will
eventually become simultaneously empty. If all the queues are empty at the same time, this means
the receiver has sent as many ACKs as the total number of packets that ever arrived at any sender
so far. As discussed above, the number of ACKs is equal to the total number of linearly independent
equations available at the receiver. In other words, the receiver has as many equations as the
unknowns, and can decode all of the packets that ever arrived at the senders.
\end{proof}

\begin{remark}
The priority-based policy requires knowledge of the arrival rates at the receiver to tune the
acknowledgement mechanism. However, if the senders' queue-length information is available at the
receiver, we can mimic the policy by Tassiulas and Ephremides \cite{LCQ} by acknowledging the
sender with the longest queue. Then, we shall not need to know the arrival rates. Achievability of
the stability region in (\ref{MAC_region}) is then a direct consequence of the results in
\cite{LCQ}.

It is worth mentioning that, if the probability of erasure is different on different links, this
scheme would still achieve the corresponding stability region based on the results of \cite{LCQ}.
\end{remark}

\section{Multiple receiver case}\label{multiple_sec}
In this section, we generalize the results of the preceding parts to the case of a single-hop
wireless erasure channel with multiple senders and receivers. Denote by $\Gamma_O(i)$ the set of
receivers that can potentially receive a packet from sender $i$, and write $\Gamma_I(j)$ for the
set of senders that can reach receiver $j$. Recall that the senders are constrained to
\emph{broadcast} the packets on all outgoing links. The goal of each sender is to deliver all the
packets in its queue to each of its neighbors. In the following we characterize the delivery time
and the stability region of the network for ZigZag decoding and compare the results with
centralized scheduling schemes.

\subsection{Delivery time characterization}
Similarly to Section \ref{delay_sec}, we study a scenario where every sender has a single packet to
deliver to all of its neighbors.

\begin{definition}
Consider a single-hop wireless erasure network with $m$ receivers and $n$ senders, each having one
packet to transmit. Define the \emph{delivery time} of receiver $j$, $T_D^{(j)}$, as the time taken
by receiver $j$ to successfully decode all packets transmitted from all senders in $\Gamma_I(j)$.
\end{definition}
%
%In a similar fashion to Section \ref{delay_sec}, we characterize the delivery time for each
%receiver by dividing it into intervals with endpoints corresponding to acknowledgements sent by
%that receiver. In particular, for every receiver $j$ define
%
%
%\begin{tabular}{lcp{2.5in}}
%\ \ $T^{(j)}_k$ &=& Time when the receiver $j$ sends the $k^{th}$ acknowledgment\\
%\vspace{.1in}
%\ \ $X^{(j)}_k$ &=& $T^{(j)}_k-T^{(j)}_{k-1}$\ \ \ \ ($T^{(j)}_0$ is assumed to be 0).\\
%\end{tabular}
%
%The delivery time of the $j^{th}$ receiver is given by
%\begin{equation}\label{eq:tdx}
%T^{(j)}_D(n)=\sum_{k=1}^{|\Gamma_I(j)|} X^{(j)}_k.
%\end{equation}
%

%
%\begin{figure}
%\centering
%  \includegraphics[width = 0.2\textwidth]{broadcast}\\
%  \caption{Single-hop wireless network model with three senders and two receivers.}\label{multi_fig}
%\end{figure}

A centralized scheduling scheme involves assigning at most one sender to each receiver so that
collisions are avoided. However, unlike the single receiver case, it is not always feasible to
assign exactly one sender to each receiver. This is due to the broadcast constraint of the senders
that may cause interference at other receivers. For example, in the configuration depicted in Fig.
\ref{fig:bipartite}, we cannot allow both of the senders to transmit simultaneously. Hence, the
delivery time for receivers 1 is affected by that of receiver 3, and it is larger than the case
where other receivers are not present. Therefore, we have
$$T^{(j)}_D \geq \frac{|\Gamma_I(j)|}{1-p}.$$

%For a random access mechanism, sender $i$ transmits its packet with probability $q_i$ until it is
%acknowledged by all the receivers it is connected to. For the single receiver case, it is natural
%to choose equal access probability $q$ for all of the senders. In contrast, the access
%probabilities for the multiple receiver case should be chosen according to the contention level at
%the receivers.

If a collision recovery method such as ZigZag decoder is implemented at the receiver, similar to
the single receiver case, every sender keeps transmitting its packet until an acknowledgement is
received from all of its neighbor receivers. If we use the acknowledgement mechanism  as in the
single receiver case, i.e., ACK any of the packets involved in a collision, then sending an
acknowledgement does not necessarily correspond to receiving an innovative equation. Moreover,
multiple ACKs may be sent to the same sender while the other senders are not acknowledged even
after decoding their packets. This is so since a sender does not stop broadcasting its packet
unless receiving ACKs from all of its neighbors. Here, we slightly modify the acknowledgement
mechanism as follows. Upon a reception at each receiver, the receiver acknowledges any of the
packets involved in the reception (collision) that have not already been acknowledged.

\begin{theorem}\label{delivery_bd_zz}
Consider a single-hop wireless erasure network with collision recovery implemented at the
receivers. The expected delivery time for each receiver $j$ is bounded from above as
$$\mathbb E\big[T^{(j)}_D\big] \le \sum_{k=1}^{|\Gamma_I(j)|} \frac 1{1-p^k} = |\Gamma_I(j)|+O(1).$$
\end{theorem}
\begin{proof}
Fix a particular receiver $j$. Suppose each sender in $\Gamma_I(j)$ stops transmitting after
receiving an ACK from $j$. By Theorem \ref{thm:correctness} all of the packets at the neighbors of
$j$ are decodable, once all of the senders in $\Gamma_I(j)$ are acknowledged, i.e., the system of
$|\Gamma_I(j)|$ equation at receiver $j$ is full rank. Therefore, even if the acknowledged packet
get retransmitted, the receiver $j$ will have a full rank system after sending $|\Gamma_I(j)|$
ACKs. Now we can divide the delivery time into intervals corresponding to ACK instances, i.e.,
\begin{equation}\label{TD_bd}
T^{(j)}_D(n) \leq \sum_{k=1}^{|\Gamma_I(j)|} X^{(j)}_k,
\end{equation}
where $X^{(j)}_k$ is the duration between sending the $(k-1)^{st}$ ACK and $k^{th}$ ACK. The
inequality could be strict if the system of equations become full rank before sending the last ACK.

Note that, at a give time slot, a new ACK is sent by receiver $j$ if and only if a collision is
received that involves at least one unacknowledged packet. Therefore, $X^{(j)}_{k+1}$ is a
geometric random variable with mean $\frac{1}{1-p^{|\Gamma_I(j)|-k}}$. Similarly to the proof of
Theorem \ref{thm:delay_zz}, the desired result is followed from plugging this into (\ref{TD_bd}).
\end{proof}

The exact characterization of the expected delivery time requires characterizing the exact decoding
process that is beyond the scope of this paper. Note that the upper bound on the expected deliver
time given by Theorem \ref{delivery_bd_zz} differs from the lower bound, $|\Gamma_I(j)|$, by only a
small constant.

\subsection{Stability region}
In this part, we study a wireless erasure network with multiple senders and receivers, where
packets arrive at sender $i$ according to the arrival  process $A_i(t)$. We assume the arrival
processes satisfy Assumption \ref{assump:arrival} for some rate $\lambda_i$. The goal is to
characterize the stability region of the system when the receivers have collision recovery
capabilities. Note that in this scenario, both broadcast and interference constraints are present,
and there are multiple broadcast sessions. We show that the cut-set bound is achievable by
combining network coding at the senders and collision recovery at the receivers.

First, let us state the outer bound given by the cut-set bound. This region is the intersection of
the stability regions given by \ref{MAC_region} for individual receivers.

\begin{theorem}\label{outer_thm}
[Outer bound] Consider a single-hop wireless erasure network with link erasure probability $p$.
Assume that packets arrive at sender $i$ with rate $\lambda_i$. For every receiver $j$, it is
necessary for stability of the system to have
\begin{eqnarray}
   \sum_{i \in S} \lambda_i &\leq& 1-p^{|S|}, \quad \textrm{for all } S \subseteq \Gamma_I(j),  \nonumber \\
    \lambda_i &\geq& 0, \quad \foral i, \label{cutset_bd}
\end{eqnarray}
where $\Gamma_I(j)$ is the set of senders in the neighborhood of receiver $j$.
\end{theorem}

\begin{proof}
Assume that the system is operating under some policy $\mathcal P$ and is stable. Hence, the Markov
chain corresponding to the queue lengths at the senders is ergodic and has a stationary
distribution. Therefore, the departure rate $\mu_i$ of the queue at sender $i$ is equal to its
arrival rate $\lambda_i$. On the other hand, by independence of the information at different
senders, the departure (transmission) rates should satisfy the following conditions given by the
cuts between each receiver $j$ and the senders over a bipartite graph:
$$   \sum_{i \in S} \mu_i \leq 1-p^{|S|}, \quad \textrm{for all } S \subseteq \Gamma_I(j),$$
which implies the desired result.
\end{proof}

Next, we present transmission and acknowledgement policies that achieve the outer bound given by
Theorem \ref{outer_thm}. The transmission policy is based on network coding, and the
acknowledgement policy is based on the notion of "seen" packets as defined in \cite{ARQ_jay}, and
is build upon a single-receiver acknowledgement policy. Let us start by some definitions and
notations.

\begin{definition}\label{ACK_policy_single}
[Single-receiver ACK policy] Consider a single-hop wireless network of a single receiver and $n$
senders. Let $\bs C \in \mathbb \{0,1\}^n, \bs Q \in \mathbb Z^n$ denote the channel state and
queue-length vectors, respectively. Define an \emph{ACK policy} as the following mapping:
$$f:  \{0,1\}^n \times \mathbb Z^n \rightarrow \{\emptyset, 1,\ldots,n\}.$$
Given the channel state  and the queue-length vectors, $f(\bs C, \bs Q)$ provides the index of at
most one sender to be acknowledged. An ACK policy is \emph{stable} if it stabilizes the queues for
any arrival rate in the stability region of the system.
\end{definition}

Note that the priority-based ACK policy given in Definition \ref{priority_policy} does not require
the queue-length information, while the ACK policy proposed by Tassiulas and  Ephremides uses the
queue-length information.

\begin{definition}\label{multiple_policy}
[Code-ACK policy] Consider a single-hop wireless erasure network. The Code-ACK policy is as
follows:
\begin{itemize}
  \item \textbf{Transmission mechanism}: Each sender transmits a random linear combination of the  packets in its queue at every time slot.
  \item \textbf{Acknowledgement mechanism}: Each receiver $j$ acknowledges the last \emph{seen} packet of
  the sender given by $f_j\big(\bs C^{(j)}(t), \bs Q^{(j)}(t)\big)$, where $f_j$ be a single-receiver ACK policy (cf. Definition
\ref{ACK_policy_single}) for $j$ when other receivers are not present, and
$$\bs C^{(j)}(t) =\{c_{ij}(t): i \in \Gamma_I(j)\},$$
$$\bs Q^{(j)}(t) = \{Q_{ij}(t): i \in \Gamma_I(j)\},$$
where $Q_{ij}(t)$ the backlog of the packets at sender $i$ not yet {seen} by receiver $j$.
\end{itemize}
\end{definition}

\begin{figure}
\centering
  \includegraphics[width = 0.4\textwidth]{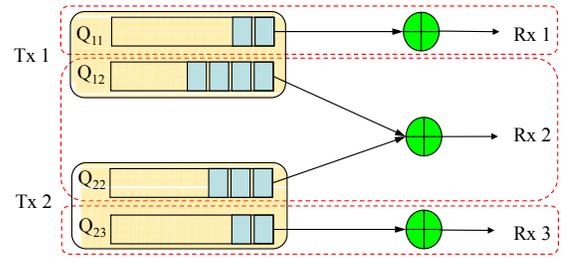}\\
  \caption{Each sender-side queue viewed as two virtual queues containing the packets not seen by the corresponding receiver. }\label{virtual_fig}
\end{figure}

\begin{theorem}\label{stability_thm}
Consider a single-hop wireless erasure network with multiple receivers all capable of collision
recovery. Assume the arrival processes at the senders satisfy Assumption \ref{assump:arrival}. The
Code-ACK policy given in Definition \ref{multiple_policy} achieves any point in the interior of the
region given by (\ref{cutset_bd}), if the single-server ACK policies $f_j$ used in Code-ACK policy
are stable. Moreover, every packet that arrives at a sender will eventually get decoded by all of
the its neighbor receivers.
\end{theorem}

\begin{proof}
Since each sender needs to deliver all of its packets to all of its neighbors, we can think of a
senders's queue as multiple virtual queues targeted for each that sender's neighbors. Each of these
virtual queues contain the packets still needed by the corresponding receiver. An arrival at the
sender corresponds to an arrival to each of its virtual queues, and an ACK from a receiver results
in dropping a packet from the virtual queue of that receiver. A packet is dropped from a sender's
original queue, if it is ACKed by all of its neighbors, in other words, if it is dropped from all
its virtual queues (See Fig. \ref{virtual_fig}). Therefore, we can relate the queue-length at
sender $i$ to those of the virtual queues as follows:
\begin{equation}\label{real_virtual_Q}
Q_i(t) \le \sum_{j \in \Gamma_O(i)} Q_{ij}(t).
\end{equation}

In the Code-ACK policy, receivers acknowledge a seen packet from a sender. Thus, the  virtual
queues at sender $i$ corresponding to receiver $j$ coincides with the packets at sender $i$ not yet
seen by receiver $j$. Moreover, upon every reception at receiver $j$, the corresponding virtual
queue of sender $f_j\big(\bs C^{(j)}(t), \bs Q^{(j)}(t)\big)$ is served. Therefore, we can isolate
each receiver $j$ and its corresponding virtual queues from the rest of the network, and treat the
isolated part as single-receiver erasure network.

By comparing the regions described in (\ref{cutset_bd}) and (\ref{MAC_region}), we observe that the
region for the  multiple-receiver case is a subset of the one for the single-receiver case. Since
$f_j$ is a stable single-receiver ACK policy for every receiver $j$, all of the virtual queues are
stable. Therefore, by (\ref{real_virtual_Q}) all of the sender-side queues are stable.

It remains to show that that all of the packets arriving at a sender are eventually decodable at
its neighbor receivers. Similarly to the proof of Theorem \ref{decodability_single}, it is
sufficient to show that for every ACK sent by receiver $j$, a degree of freedom (innovative packet)
is received at receiver $j$. If this is the case, by stability of the virtual queues corresponding
to receiver $j$, they all eventually become empty and there are as many degrees of freedom at the
receiver as there are unknowns. Hence, every packet arrived at the senders in $\Gamma_I(j)$ are
decodable.

Now, we prove the above claim. Let receiver $j$ send and ACK to sender $i$ at the end of slot $t$.
First, we observe that the link between $i$ and $j$ should be connected during slot $t$, and
$Q_{ij}(t) >0$. Sender $i$ broadcasts a random linear combination of the packets in its queue which
include the packets in the virtual queue $Q_{ij}$. If the field size is large enough, we can assume
that the coefficients corresponding to at least one of the packets in virtual queue $Q_{ij}$ is
nonzero. Hence, the reception at receiver $j$ at time slot $t$ should have involved a packet from
sender $i$ that was not seen by receiver $j$. Since all decoded packets are seen \cite{ARQ_jay},
the collision at receiver $j$ at time $t$ involves a packet that is not yet decoded, and hence, it
is a new degree of freedom (innovative reception).
\end{proof}

\section{Conclusions}\label{conclusion_sec}
In this paper, we have studied the delay and throughput performance of collision recovery methods,
e.g. ZigZag decoding \cite{ZigZag}, for a single-hop wireless erasure network. Using an algebraic
representation of the collisions allowed us to view  receptions at a receiver as linear
combinations of the packets at the senders. The algebraic framework provides alternative collision
recovery  methods and generalizations for the case when the transmitted packets are themselves
coded versions of the original packets.

 We have focused on two situations -- the completion time for all of the senders to deliver a
single packet to their neighbor receivers, and the rate region in the case of streaming arrivals.
We show that the completion time at a receiver with collision recovery is at most by a constant
away from the degree of that receiver which is the ultimate lower bound in this setup. For the
streaming case, we present a decentralized  acknowledgement mechanism that could serve as an
ARQ-type mechanism for achieving the capacity of a wireless erasure network when both broadcast and
interference constraints are present. Our conclusion is that collision recovery approach allows
significant improvements upon conventional contention resolution approaches in both the completion
time as well as the rate region, while not requiring any coordinations among the senders.

\end{document}